\documentclass[12pt]{article}

\usepackage[reqno]{amsmath} 
\usepackage{amssymb}
\usepackage{amstext}
\usepackage{amsthm}
\usepackage{amsfonts}
\usepackage{mathtools}
\usepackage{enumerate}  
\usepackage{dsfont}
\usepackage[utf8]{inputenc}

\usepackage[backend=bibtex,style=numeric,doi=false,isbn=false,url=false,hyperref=true,backref=true]{biblatex}

\bibliography{staggeredcontinuous.bib}

\usepackage[pdftex]{graphicx} 
\usepackage{hyperref} 

\theoremstyle{plain}

\swapnumbers

\newtheorem{theorem}{Theorem}[section]

\theoremstyle{definition}

\usepackage[dvipsnames]{xcolor}
\usepackage{soul,xcolor}
\usepackage{notation}

\title{
Discretization of continuous-time quantum walks via the staggered model with Hamiltonians}
\author{  
Gabriel Coutinho \footnote{Acknowledges the support of grants FAPESP 15/16339-2 and FAPESP 03447-6.} \\ \small Dep. of Computer Science \\ \small IME-USP \\ \small São Paulo, SP, Brazil \\ \small \texttt{coutinho@ime.usp.br} 
\and 
Renato Portugal \footnote{Acknowledges the support of CNPq grant 474143/2013-9.} \\ \small National Lab of Scientific Computing \\ \small LNCC \\ \small Petr\'opolis, RJ, Brazil \\ \small \texttt{portugal@lncc.br}}

\date{  \today  }

\begin{document}

\maketitle

\begin{abstract}
We characterize a close connection between the continuous-time quantum-walk model and a discrete-time quantum-walk version, based on the staggered model with Hamiltonians in a class of Cayley graphs, which can be considered as a discretization of continuous-time quantum walks. This connection provides examples of perfect state transfer and instantaneous uniform mixing in the staggered model. On the other hand, we provide some more examples of perfect state transfer and instantaneous uniform mixing in the staggered model that cannot be reproduced by the continuous-time model. 
\end{abstract}

\newpage
\section{Introduction}

Quantization versions of classical random walks have been actively studied for the past two decades after the seminal papers~\cite{ADZ93,FarhiGutmann,SzegedyQWalk}. In this work we focus on extending some properties of the recent-proposed model of staggered quantum walks~\cite{PortugalSantosFernandesStaggered,Por16,PortugalOliveiraStaggered}. This model is endowed with many interesting features, remarkably, it provides a generalization of the Szegedy's quantum walk model \cite{SzegedyQWalk} and exactly simulates the instances of flip-flop coined quantum walks that employ the Hadamard or Grover coins~\cite{PortugalSzegedyStaggered}. An implementation of this model in the class of triangle-free graphs has also been recently proposed \cite{MoqadamOliveiraPortugalStaggeredImplementation}.

Connections between continuous-time quantum walks and discrete-time coined quantum walks have been extensively analyzed~\cite{Str06,ChildsRelContDiscreQWalks,DB15,PP16} without providing a discretization procedure. In this paper, we show that in a special class of graphs and under certain conditions, the staggered quantum walk model is able to provide a natural discretization of a continuous-time quantum walk when the adjacency matrix can be written as a sum of commuting matrices related to a graph factorization.

Perfect state transfer has been demonstrated in spin chains~\cite{Bos07,Kay10} and can be mathematically modeled using the continuous-time quantum walk on graphs~\cite{KendonTamon,CoutinhoPhD}. In this paper, we show that there is a large class of graphs on which the staggered quantum walk allows perfect state transfer. Part of this class is obtained by discretizing the examples from the continuous-time quantum walk. The remaining part shows that the staggered model allows for certain interesting features in terms of perfect state transfer that the continuous case does not possess. We then proceed to show that related phenomena, such as uniform mixing and periodicity, are also present.

The remainder of the paper is organized as follows. In section~\ref{sec:SQW}, we review the staggered quantum walk model. In section~\ref{sec:disc_of_CTQW}, we present a discretization of continuous-time quantum walks. In section~\ref{sec:tess-factor}, we describe a especial kind of tessellation on Cayley graphs. In section~\ref{sec:PST}, we provide a large class of examples of staggered quantum walks on Cayley graphs of abelian groups that admits perfect state transfer. We also discuss staggered quantum walks on Cayley graphs of non-abelian groups. In section~\ref{sec:IUM}, we provide a large class of examples of staggered quantum walks on Cayley graphs that admits instantaneous uniform mixing.  In section~\ref{sec:conc}, we draw our conclusions.

\section{Staggered quantum walks with \\ Hamiltonians}\label{sec:SQW}

We introduce some of the nomenclature needed along the text. We model a quantum walk on a finite graph by associating particle positions with vertices and hopping directions with edges.  In this fashion, the quantum system is a graph, typically denoted as $X = (V,E)$, where $V$ is the set of vertices and $E \subseteq \binom{V}{2}$ is the set of edges.

A \textit{tessellation} is a partition of $V$ into tiles $C_1,...,C_k$ such that if any two vertices $a$ and $b$ belong to a tile, then $\{a,b\} \in E$. Given a tessellation, a unitary operator $H$ is defined by attaching to each tile $C_i$ a unitary vector $u_i$ such that the support of $u_i$ is precisely the set of vertices in $C_i$, and then making
\[H = 2\sum_{i = 1}^k  |u_i\rangle \langle u_i | - I.\]
This way, $H$ is forcibly unitary and Hermitian, that is, a reflection operator.

The staggered quantum walk model proposed in \cite{PortugalSantosFernandesStaggered} consists in defining (at least) two tessellations, say $T_1$ and $T_2$, such that each edge of the graph belongs to at least one of them, along with the unitary vectors associated to each. The evolution operator is therefore given by
\[U = H_1 H_2.\]
A variation of this model more suitable for practical implementation was introduced in \cite{PortugalOliveiraStaggered}. Again, provided tessellations with their corresponding unitary matrices, and two parameters $\theta_1$ and $\theta_2$, the evolution operator is given by
\begin{align} U = \exp(\ii \theta_1 H_1) \exp(\ii \theta_2 H_2). \label{eq1} \end{align}
The staggered model with Hamiltonians takes advantage of the dual nature of reflection operators; they can either be used as propagators or Hamiltonians. Henceforth in this paper, we consider the model described in equation (\ref{eq1}) and its generalization with more than two tessellations.

The greatest difficulty in dealing with this model comes from the fact that the spectrum of $U$ has little to do with the structure of the underlying graph in the general case. In fact, other than unitary, little can be said about $U$ at all. For this reason, we enforce extra properties. These might come in two flavours: to restrict the choices for the unitary vectors used at each tessellation, and to see that the structure of the tessellations are somehow dependent on each other. In this direction, we focus on the following case
\begin{itemize}
\item The unit vectors attached to each tile in a tessellation are real and uniform, that is, if the tile contains $\gamma$ vertices, the vector has entries equal $1/\sqrt{\gamma}$ at each vertex in the tile, and $0$ otherwise.
\item The matrices $H_i$, corresponding to each tessellation $T_i$, commute.
\end{itemize}
Perhaps quite surprisingly, there is a large class of graphs satisfying both properties above. 

\section{A discretization of continuous-time\\  quantum walks on special graphs}
\label{sec:disc_of_CTQW}

Given a graph $X$, a \textit{tessellation covering} of $X$ consists in a set of tessellations such that each edge of $X$ belongs to at least one tessellation. A tessellation covering is a \textit{factorization} if each edge belongs to precisely one tessellation. A tessellation covering is \textit{uniform} if each tessellation induces a partition into cliques of the same size. Note that although all graphs admit a tessellation factorization, in most cases these tessellations will not be uniform. A typical example of a uniform tessellation factorization is a factorization into perfect matchings, also known as a $1$-factorization.

Assume $X$ has a uniform tessellation covering into $k$ tessellations, let $A_i$ denote the adjacency matrix of the subgraph induced by the $i$th tessellation and let $H_i$ denote the unitary matrix obtained from each tessellation using the real and positive uniform superposition. Let $\gamma_i$ denote the size of the cliques in the $i$th tessellation, and define $\gamma = \sum_{i = 1}^k \gamma_i$.
 
Note that
\begin{align}\label{eq2} \gamma_i H_i = 2 A_i + (2 -\gamma_i) I. \end{align}
Now assume $\{A_1,...,A_k\}$ are commuting matrices. This means that $\{H_1,...,H_k\}$ are also commuting matrices, hence
\[U^T = \left(\prod_{i = 1}^k \exp (\ii \theta_i H_i) \right)^T = \prod_{i = 1}^k \exp (\ii \theta_i H_i)^T,\]
and the quantum dynamics can be analysed locally at each $H_i$. 
Moreover, if the covering is a factorization, then $A(X) = \sum_{i = 1}^k A_i$. Thus it follows that
\begin{align*}
U & = \prod_{i = 1}^k \exp (\ii \theta_i H_i) \\
& = \exp \left( \ii \sum_{i = 1}^k \theta_i \frac{2 A_i + (2 -\gamma_i) I}{\gamma_i} \right) \\
& = \exp \left( \ii \sum_{i = 1}^k \frac{\theta_i}{\gamma_i} \left( 2 - \gamma_i \right) \right) \exp \left( 2 \ii \sum_{i = 1}^k \frac{\theta_i}{\gamma_i} A_i \right),
\end{align*}
and upon choosing $\theta_i$ to be a constant multiple of $\gamma_i$, say $\theta_i = \theta \gamma_i$, it follows that, for any $T \in \Z_+$,
\[U^T = \exp(\ii \theta (2k - \gamma ) T ) \exp(2 \ii \theta T A).\]

This shows that, when a uniform tessellation factorization exists, it is possible to use the staggered quantum walk model from \cite{PortugalOliveiraStaggered} to emulate a continuous quantum walk model, that is, $\{U^0 , U^1 , U^2 , ...\}$ provides a discretization of $\exp(\ii t A)$ in units of $2 \theta$, up to a global phase.

This immediately raises the question of how common such structures are. We devote the next section to study a broad class of examples.

\section{Tessellation-factorizations in Cayley graphs}\label{sec:tess-factor}
\newcommand{\G}{\mathcal{G}}
Given a group $\G$ and a subset $\CC$ of group elements which does not contain the identity $\texttt{id}$ and is closed under taking the inverse, the Cayley graph $\text{Cay}(\G,\CC)$ is the graph whose vertex set is $\G$, and two vertices $g$ and $h$ are adjacent if $gh^{-1} \in \CC$. The set $\CC$ is typically called the connection set.

Consider a staggered quantum walk on the Cayley graph $X = \text{Cay}(\G,\CC)$, where $\G$ is a finite abelian group and $\CC$ is a connection set with the following properties
\begin{itemize}
\item[($1$)] $\G=\langle \CC \rangle$, and
\item[($2$)] $g \in \CC \implies g^k \in \CC$ for  $0<k<\textrm{ord}(g)$, 
\end{itemize}
where $\textrm{ord}(g)$ is the order of $g$. Property~(1) ensures that the graph is connected. Property~(2) implies that if $g\in \CC$ then set $\{g^k : 0\le k<\textrm{ord}(g)\}$ is a clique in $X$ and can be considered as a tile of some uniform tessellation of $X$. 

If  $\CC$ is partitioned as $\CC_1  \cup ... \cup \CC_k$, such that each $\CC_i \cup \texttt{id}$ is a subgroup of $\G$, then this partition induces a uniform factorization of $X$ with $k$ tessellations in the following way. The cosets of $\CC_i \cup \texttt{id}$ in $\G$ are cliques (tessellation tiles) and the union of those cliques is the $i$th tessellation. The tessellation-factorization has no edges in the tessellation intersection and defines a staggered quantum walk on $X$ with $k$ tessellations.  

As before, if $A_i$ is the adjacency matrix of the subgraph induced by the $i$th tessellation, that is, if
\[A_i = A ( \text{Cay}(\G,\CC_i)), \]
then, as the group is abelian, these matrices will commute. Defining $H_i$, $\gamma_i$ and $\theta_i = \theta/\gamma_i$ and $U$ as in the previous section, then powers of $U$ provide a discretization of the continuous-time quantum walk $\exp(\ii t A)$. As a result, phenomena such as perfect state transfer or uniform mixing in continuous time quantum walks (\cite{AdaChanComplexHadamardIUMPST,GodsilCheungPSTCubelike,BasicCirculant,GodsilMullinRoy}) can also be observed in (coinless) discrete time quantum walks.

However, we will see in the examples below that we are able to impose more control over the quantum dynamics if we vary the $\theta_i$s and if we do not require the tessellations to provide a factorization. As a result, we can construct new examples of perfect state transfer, uniform mixing, and related phenomena.

\section{Perfect state transfer}\label{sec:PST}

For a detailed introduction in the topic of perfect state transfer in graphs according to the continuous-time quantum walk model, we refer the reader to \cite[Chapter 2]{CoutinhoPhD}, or \cite{GodsilStateTransfer12}.

\begin{theorem}\label{Theo:PST}
Let $X = \text{Cay}(\G,\CC)$. Suppose $\CC$ is partitioned as $\CC_1,...,\CC_k$, each $\CC_i \cup \texttt{id}$ a subgroup, and at least one ${\CC_i}$ has exactly one element. Define Hermitian commuting matrices $H_i$ to each $\CC_i$ as in Eq.~(\ref{eq2}). Define
\[U = \prod_{i = 1}^k \exp(\ii \theta_i H_i).\]
By conveniently selecting the $\theta_i$, $U^T$ admits perfect state transfer at any chosen time $T$.
\end{theorem}
\begin{proof}
Suppose ${\CC_1}$ has only one element, which must be of order 2. Denote by $\{A_1,...,A_k\}$ the adjacency matrices of each tessellation induced by the cosets of $\CC_i \cup \texttt{id}$ in $\G$. Since $A_1$ corresponds to $\CC_1$, which comprises one order-2 element, it follows that $A_1$ is a perfect matching. Let $\gamma_i$ be the size of the cliques in $A_i$, that is, $\gamma_i$ is the cardinality of $\CC_i \cup \texttt{id}$. It follows that the distinct eigenvalues of each $A_i$ are $\gamma_i - 1$ and $-1$, hence these are integral graphs, and thus periodic. This means that if $2 T \theta_i/\gamma_i$ is an even multiple of $\pi$, then $\exp(2 \ii (\theta_i / \gamma_i) T A_i) = I$. Upon choosing $\theta_1$ such that $T \theta_1 = \pi/2$, it follows that
\[U^T = \alpha A_1,\]
where $\alpha$ is complex number of absolute value $1$, depending on the $\theta_i$s and $\gamma_i$s. Hence perfect state transfer occurs in this case. This construction can be carried out in any abelian group of even order, as an element of order 2 necessarily exists.
\end{proof}

If more elements of order $2$ are available in the connection set, say $g_1,...,g_\ell$, perfect state transfer between the vertices corresponding to $\texttt{id}$ and any product of the $g_i$s can be manufactured by conveniently selecting the $\theta_i$s for a given $T$. 

In Figure \ref{fig1}, we depict a Cayley graph for $\Z_2 \times \Z_2 \times \Z_3$, with connection set ${\CC}=\{(1,0,0) , (0,1,0) , (0,0,1) , (0,0,2)\}$. Let $H_1$, $H_2$ and $H_3$ be the unitary Hermitian matrices obtained from the uniform superposition on the tessellations given by ${\CC_1}=\{(1,0,0)\}$, ${\CC_2}=\{(0,1,0)\}$ and ${\CC_3}=\{(0,0,1) , (0,0,2)\}$ respectively. The corresponding tessellations consist of horizontal edges (including curved ones), vertical edges and edges in the triangles, respectively. Define $U$ as
\[U = \exp(\ii \theta_1 H_1) \exp(\ii  \theta_2 H_2) \exp(\ii \theta_3 H_3).\]
If $\theta_1 = \theta_2 = \theta$, and $\theta_3 = 12 \theta$, then perfect state transfer occurs between $a_i$ and $b_i$ at time $T = \pi /2\theta$. By making $\theta_2 = 2 \theta$, then perfect state transfer occurs between $a_i$ and $c_i$.
\begin{figure}
\begin{center}\includegraphics[]{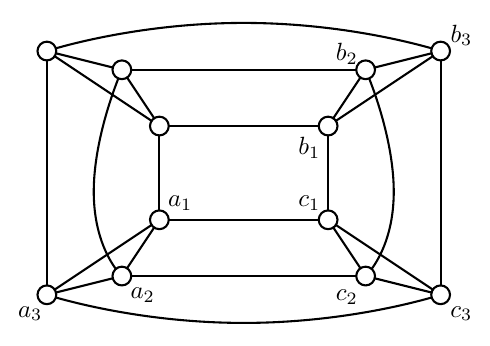}\end{center}
\caption{Example of Cayley graph admitting perfect state transfer in the staggered quantum walk model.}
\label{fig1}
\end{figure}

In any example of perfect state transfer at time $T$ based on Theorem~\ref{Theo:PST}, the quantum walk will be periodic with period $2T$, because $U^{2T}=I$ modulo a global phase. It is possible to have periodicity without perfect state transfer if $\theta$ is a rational multiple of $\pi$ and if the perfect state transfer based on Theorem~\ref{Theo:PST} would be produced at a rational (non-integer) number of steps, say $T=m/n$, where $m,n$ are coprime integers. In this case, the staggered quantum walk does not have perfect state transfer, but it is periodic with period $m$.

\subsection{Example with an edge in the tessellation \\ intersection}

It is not necessary for the connection set to be partitioned into disjoint subsets. As long as each element in the connection set belongs to a subgroup entirely contained in the connection set, a set of tessellations containing all edges can be defined. However in this case, we no longer have a tessellation factorization, but simply a (uniform) tessellation covering.

Here is an example. Let $\G= \Z_2\times \Z_4$ and \[\CC=\{(1,0),(0,1),(0,2),(0,3),(1,1),(1,3)\}.\] Suppose that $\CC_1=\{(1,0)\}$, $\CC_2=\{(0,1),(0,2),(0,3)\}$, and $\CC_3 = \{ (1,1), (0,2),\allowbreak (1,3) \}$. Notice that $\CC_2\cap\CC_3$ is nonempty, which means that their corresponding tessellations will have an edge in the intersection. Figure \ref{fig2} depicts the three tessellations separately.

\begin{figure}
\begin{center}\includegraphics[]{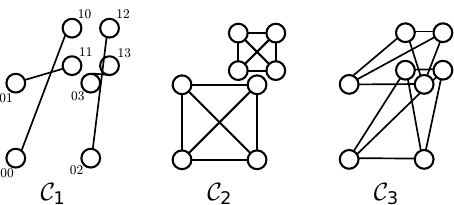}\end{center}
\caption{Example of a tessellation covering with edges in the intersection in which perfect state transfer still occurs.}
\label{fig2}
\end{figure}

If $\theta_1 = \theta$ and $\theta_2 = \theta_3 = 4 \theta$ then perfect state transfer occurs between the vertices connected by the edges that belong to $\CC_1$ at time $T = \pi /2\theta$, for instance, from vertex $(1,2)$ to  $(0,2)$. Notice again that we have to choose $\theta$ so that $T$ is a positive integer. Otherwise, if $\pi /2\theta$ is a rational number $m/n$ with $n>1$, there is no perfect state transfer, but the walk is periodic modulo a global phase with period $m$.

In this example, edges associated with generator $(0,2)$ in the Cayley graph $\G$ belong to tessellations $\CC_2$ and $\CC_3$. Examples with edges in the tessellation intersection are interesting because staggered quantum walks based on those tessellations do not reduce to the coined or Szegedy's models~\cite{Por16}. 

\subsection{Non-abelian case}

So far we have shown examples using Cayley graphs of abelian groups. Examples with non-abelian groups can be provided as well after taking an extra care of checking that the matrices $H_i$ used in Theorem~\ref{Theo:PST} commute. As a guide to choose connection sets, we have the following result.

\begin{theorem} \label{prop1}
Let $A_1$ and $A_2$ be the adjacency matrices of $\text{Cay}(\G,\CC_1)$ and $\text{Cay}(\G,\CC_2)$, respectively, where $C_1=\CC_1 \cup \texttt{id}$ and $C_2=\CC_2 \cup \texttt{id}$ are subgroups of $\G$.
$A_1$ and $A_2$ commute if and only if $C_1\cdot C_2 = C_2\cdot C_1$.
\end{theorem}
\begin{proof}
First note that $A_1$ and $A_2$ commute if and only if $(A_1 + I)$ and $(A_2 + I)$ commute. Now let $g$ and $h$ be elements of $\G$, and hence also vertices of $\text{Cay}(\G,\CC_1)$ and $\text{Cay}(\G,\CC_2)$.

By the definition of adjacency matrix of a Cayley graph, $(A_1+I)_{g a}=1 \Leftrightarrow a \in C_1\cdot g$ and $(A_2+I)_{g a}=1 \Leftrightarrow a \in C_2\cdot g$. Then, 
\[[(A_1 + I)(A_2+I)]_{gh} = \sum_{a\in \G}(A_1+I)_{g a}(A_2+I)_{a h} = \left|(C_1\cdot g)\cap (C_2 \cdot h)\right|.\] 
Likewise, $\sum_{a\in \G}(A_2+I)_{g a}(A_1+I)_{a h} = \left|(C_2\cdot g)\cap (C_1\cdot h)\right|$. Moreover, note that if $\left|(C_1\cdot g)\cap (C_2 \cdot h)\right| \neq 0$, then there is $a_1 \in C_1$ and $a_2 \in C_2$ such that $a_1 g = a_2 h$. Then $C_1 \cdot g = C_1 \cdot (a_1^{-1} a_2 h) = C_1 \cdot a_2 h$. Thus,
\begin{align}\left|(C_1\cdot g)\cap (C_2 \cdot h)\right| & = \left|(C_1\cdot a_2 h)\cap (C_2 \cdot h)\right| = \left|(C_1\cdot a_2 )\cap C_2 \right| \nonumber\\ & = \left|C_1\cap C_2 \cdot a_2^{-1} \right|= \left|C_1\cap C_2 \right|. \label{eq3}\end{align}
Likewise, if $\left|(C_1\cdot h)\cap (C_2 \cdot g)\right| \neq 0$ then $\left|(C_1\cdot h)\cap (C_2 \cdot g)\right| = \left|C_1\cap C_2 \right|$. We now proceed to show that $C_1\cdot C_2=C_2\cdot C_1$ if and only if, for all $g,h\in \G$, $\left|(C_1\cdot g)\cap (C_2 \cdot h)\right|=\left|(C_2\cdot g)\cap (C_1\cdot h)\right|$, which because of (\ref{eq3}) is equivalent to showing that these intersections are either both empty or both non-empty.

Suppose that $\left|(C_1\cdot g)\cap (C_2 \cdot h)\right| \neq 0$ and $\left|(C_1\cdot h)\cap (C_2 \cdot g)\right| = 0$ (the opposite case is analogous). There is $a_1 \in C_1$ and $a_2 \in C_2$ such that $a_1 g = a_2 h$. Then clearly $hg^{-1} \in C_2\cdot C_1$. However, if $hg^{-1} \in C_1\cdot C_2$ then $(C_1\cdot h)\cap (C_2 \cdot g)$ is also non-empty. Hence, $C_2\cdot C_1 \neq C_1\cdot C_2$. This shows that $C_2\cdot C_1 = C_1\cdot C_2$ implies $A_1 A_2 = A_2 A_1$.

Now assume that $|C_1\cdot C_2| > |C_2\cdot C_1|$ (the opposite case is analogous). Let $a_1 \in C_1$ and $a_2 \in C_2$ be such that $a_1 a_2 \notin C_2 \cdot C_1$. It follows that $\left|(C_1 \cdot a_2^{-1} )\cap (C_2 \cdot a_1)\right| = 0$, then $(A_2+I)(A_1+I)_{a_1 a_2^{-1}} = 0$, however
\[(C_1 \cdot a_1 )\cap (C_2 \cdot a_2^{-1}) = C_1 \cap C_2,\]
which contains at least the identity element. Thus ${(A_1+I)(A_2 +I)_{a_1 a_2^{-1}} \neq 0}$. This shows that $A_1 A_2 = A_2 A_1$ implies $C_1 \cdot C_2 = C_2 \cdot C_1$.
\end{proof}

If $\G$ is non-abelian, a staggered quantum walk on $\text{Cay}(\G,\CC)$ admits PST if the connection set $\CC$ can be written as $\CC=\CC_1\cup ...\cup\CC_k$ so that  $C_i=\CC_i \cup \texttt{id}$ and $C_j=\CC_j \cup \texttt{id}$ are pairwise commuting subgroups of $\G$ for all $i$ and $j$ and at least one $\CC_i$ has exactly one element.

Theorem~\ref{prop1} poses a strong restriction on the possible choices of connection sets. It seems that after taking into account this restriction any instance of perfect state transfer with a non-abelian group $\G$ can be reproduced by using some abelian groups of the same order in place of $\G$. Notice that in our construction it does not matter whether $C_i$ is abelian or not.  

\section{Instantaneous uniform mixing}\label{sec:IUM}

We say that the discrete time quantum walk described by the $n \times n$ evolution operator $U$ admits instantaneous uniform mixing at time $T$ if $U^T$ is a flat complex matrix, that is, all its entries have absolute value equal to $1/\sqrt{n}$. In continuous time quantum walks, instantaneous uniform mixing has been studied in, for instance,  \cite{AdaChanComplexHadamardIUMPST}, \cite{GodsilMullinRoy}, \cite{TamonAdamczakUniformMixingCycles}, \cite{TamonAhmadiUniformMixingCirculants}, \cite{GodsilZhanIUMCayley}. In \cite{TamonAhmadiUniformMixingCirculants}, it was shown the instantaneous uniform mixing occurs in a complete graph $K_n$ if and only if $n=2,3,4$.

\begin{theorem}\label{Theo:IUM}
	Let $X = \text{Cay}(\G,\CC)$ be a Cayley graph for an abelian group. Suppose $\CC$ is partitioned as $\CC_1,...,\CC_k$, each $\CC_i \cup \texttt{id}$ a group of order at most $4$, and $\prod_{i = 1}^k (|\CC_i| + 1) = |\G|$. Let $A_i$ be the adjacency matrix of each $\text{Cay}(\G,\CC_i)$, and define Hermitian matrices $H_i$ to each $\CC_i$ as in Eq.~(\ref{eq2}). Thus, for a suitable choice of $\theta_i$s,
	\[U = \prod_{i = 1}^k \exp(\ii \theta_i H_i)\]
	admits instantaneous uniform mixing at any chosen time $T$.
\end{theorem}

\begin{proof}
	First note that because the groups $\CC_i \cup \texttt{id}$ have order at most $4$, the matrices $A_i$ will be a collection of disjoint edges, disjoint triangles, or disjoint complete graphs of size $4$, denoted by $K_4$. Each of these components admits instantaneous uniform mixing. For each of these graphs define $H_i$ as in equation (\ref{eq2}), thus the matrices $H_i$ admit uniform mixing at times $\pi/4$, $\pi/3$, or $\pi/2$, respectively. This is because the matrix $H_i$ is the direct sum of $n/\gamma_i$ Grover operators of dimension $\gamma_i$ and the absolute value of any diagonal entry of $\exp(\ii \theta_i H_i)$ is equal to a nonzero nondiagonal entry if and only if $|\sin\theta_i|=\sqrt n/2$, which indeed has real solutions only for $n\le 4$. Therefore, selecting
	\begin{equation}
		\theta_i = 
		\begin{cases}
			\frac{\pi }{ 4T}, &\quad\text{if }A_i\text{ is a collection of disjoint edges,}\\
			\frac{\pi}{3T}, &\quad\text{if }A_i\text{ is a collection of disjoint triangles,} \\
			\frac{\pi }{ 2T}, &\quad\text{if }A_i\text{ is a collection of  disjoint }K_4\text{s,} \ 
		\end{cases}
	\end{equation}
	it follows that, for all $i$, $\exp(\ii \theta_i H_i T)$ contains entries that are either $0$ or complex numbers of absolute value equal to $1/\sqrt{\gamma_i}$, where $\gamma_i$ is the order of $\CC_i \cup \texttt{id}$. Moreover, an entry $(a,b)$ is non-zero if and only if $a^{-1}\cdot b \in \CC_i \cup \texttt{id}$. Now we proceed to show that $U$ as defined in the statement admits instantaneous uniform mixing.
	
	Given a subgroup $\S$ of $\G$, let $M_\S$ be a matrix of dimension $\left|\G\right|$ such that entry $(a,b)$ is a complex number of absolute value $1/\sqrt{|\S|}$ if $a^{-1}\cdot b \in \S$, $0$ otherwise. Let $\T$ be another subgroup and define $M_\T$ analogously. Let $\RR = \S \cdot \T$.
	\begin{itemize}
		\item[] Claim 1: Each element of $\RR$ can be uniquely written as $a_1 a_2$ with $a_1 \in \S$ and $a_2 \in \T$ if and only if $|\RR| = |\S| |\T|$.  In fact, $|\S| |\T| \geq |\RR|$, and equality holds if and only if there are no duplicates in the $\S \cdot \T$. Note that both conditions are equivalent to the intersection of $\S$ and $\T$ being trivial. \\
		\item[] Claim 2: If the conditions in Claim 1 hold, then the product $M_\RR = M_\S \, M_\T$ is a matrix of dimension $\left|\G\right|$ such that entry $(a,b)$ is a complex number of absolute value $1/\sqrt{|\RR|}$ if $a^{-1}\cdot b \in \RR$, $0$ otherwise. In fact, note that
		\[(M_\S \, M_\T)_{\texttt{id},a} = \sum_{b \in \G} (M_\S)_{\texttt{id},b} (M_\T)_{b,a}.\]
		If $a \in \RR$ then, from Claim $1$ and the definitions of $M_\S$ and $M_\T$, there is a unique $b \in \S$ such that $a =b \cdot (b^{-1}\cdot a)$, and $(b^{-1}\cdot a) \in \T$. That is, a unique $b \in \G$ such that $(M_\S)_{\texttt{id},b} \neq 0 \neq (M_\T)_{b,a}$. Moreover, the product of these entries will be a complex number of absolute value $1/\sqrt{|\S|\, |\T|}$. From Claim 1, this is $1/\sqrt{|\RR|}$. On the other hand, if $a \not\in \RR$, there is no $b$ such that $(M_\S)_{\texttt{id},b}\neq 0 \neq (M_\T)_{b,a}$, otherwise $b \cdot (b^{-1}\cdot a) = a \in \S \cdot \T = \RR$. \end{itemize} 
	Recall the hypothesis that the product of the cardinalities of the subgroups $\CC_i \cup \texttt{id}$ is equal to $|\G|$. Because these subgroups generate $\G$, we have
	\[|\G| = \left|\prod_{i = 1}^{k} (\CC_i \cup \texttt{id}) \right| \leq \prod_{i = 1}^{k} |\CC_i \cup \texttt{id}| = |\G|.\]
	Thus each element of $\G$ is written uniquely as a product of elements in each $\CC_i \cup \texttt{id}$. Therefore we can recursively split the product into parts, and using Claims 1 and 2, it follows that $U$ is a matrix in which all entries are complex numbers of absolute value $1/\sqrt{|\G|}$.
\end{proof}

Note that the hypothesis that $|\G| = \prod (|\CC_i| + 1)$ is equivalent to saying that $\G \cong \bigotimes \CC_i \cup \texttt{id}$. The fact that these factors are also determining the connection set implies that the theorem generates, up to isomorphism, a unique Cayley graph admitting uniform mixing for each abelian group of order $2^k 3^m$ that has no element of orders $8$ or $9$.

We note however that this hypothesis is not strictly necessary to achieve uniform mixing as described, given that we do know examples in which the connection set is denser and yet uniform mixing occurs. For example, any hypercube in which antipodal points have been made adjacent \cite{AdaChanComplexHadamardIUMPST}. However we also know examples where the hypothesis fails and uniform mixing cannot occur.

Finally, recall that if $n > 4$ then $K_n$ does not admit instantaneous uniform mixing in the continuous-time model. The arguments in the proof above also show that if $\CC$ is partitioned into subgroups that generate each element of $\G$ uniquely and at least one of these subgroups has an element of order larger than $4$, then there is no choice of $\theta_i$s that would allow for instantaneous uniform mixing to occur in $\Cay(\G,\CC)$ using the staggered quantum walk model.

\section{Conclusions}\label{sec:conc}

The main results of this work are Theorems~\ref{Theo:PST} and~\ref{Theo:IUM}, which describe a large class of examples of staggered quantum walks on Cayley graphs of abelian groups that admit either perfect state transfer or instantaneous uniform mixing. We have also discussed related topics such as periodicity and extensions using non-abelian groups.

Throughout the text, we presented a close connection between the continuous-time and the staggered quantum walk model. Under some assumptions, when we can cover the graph with a uniform tessellation, the staggered model provides a natural discretization of continuous-time quantum walks.

\printbibliography

\end{document}